\documentclass[aps,prl,reprint]{revtex4-1}

\usepackage{soul,color}
\usepackage{xcolor}
\usepackage{graphicx}

\usepackage{amsfonts}
\usepackage{amsthm}

\usepackage{amssymb}
\usepackage{amsmath}
\usepackage{dcolumn}
\usepackage{physics}
\usepackage{float}
\usepackage{bm}
\usepackage{verbatim}
\usepackage{braket}
\usepackage{bbm}

\newcommand{\Tsym}{\mathcal{T}_{\rm sym}(n,m)}
\newcommand{\Tcyc}{\mathcal{T}_{\rm cyc}(n,m)}

\usepackage[colorlinks,linkcolor=blue,anchorcolor=blue,citecolor=blue,urlcolor=blue]{hyperref}

\newtheorem{theorem}{Theorem}

\usepackage{mathtools}
\DeclarePairedDelimiter\ceil{\lceil}{\rceil}
\DeclarePairedDelimiter\floor{\lfloor}{\rfloor}

\begin{document}
\title{Quantum entanglement enables single-shot trajectory sensing \\ for weakly interacting particles}
\author{Zachary E. Chin} \email{zchin@mit.edu}\affiliation{Department of Physics, Massachusetts Institute of Technology, Cambridge, Massachusetts 02139, USA}
\author{David R. Leibrandt}\affiliation{Department of Physics, University of California, Los Angeles, CA 90095, USA}
\author{Isaac L. Chuang}\affiliation{Department of Physics, Center for Ultracold Atoms, and Research Laboratory of Electronics, \\ Massachusetts Institute of Technology, Cambridge, Massachusetts 02139, USA} 
\date{\today}
\begin{abstract}
    Sensors for mapping the trajectory of an incoming particle find important utility in experimental high energy physics and searches for dark matter. For a quantum sensing protocol that uses projective measurements on a multi-qubit sensor array to infer the trajectory of an incident particle, we establish that entanglement can dramatically reduce the particle-qubit interaction strength $\theta$ required for perfect trajectory discrimination. Within an interval of $\theta$ above this reduced threshold, any unentangled sensor requires $\Theta(\log(1/\epsilon))$ repetitions of the protocol to estimate a previously unknown particle trajectory with $\epsilon$ error probability, whereas an entangled sensor can succeed with zero error in a single shot. Furthermore, entanglement can enhance trajectory sensing in realistic scenarios where $\theta$ varies continuously over the sensor qubits, exemplified by a Gaussian-profile laser pulse propagating through an array of atoms. 
\end{abstract}

%
%
%
%

\maketitle

\textit{Introduction.---}A particle's trajectory through space and time is a fingerprint concealing its unique history and some of its most important properties. For example, the momentum and charge of high energy particles produced in colliders is revealed from the curvature of their paths through a magnetic field \cite{atlas,cms}. Additionally, the cosmic origins of particles such as muons, neutrinos, and possible dark matter candidates can be traced from snapshots of their motion taken with detectors such as bubble chambers \cite{bubbles}. In biology and chemistry, positron emission tomography localizes tumors along a line traversed by emitted gamma photons \cite{imaging}, and mass spectrometers determine molar masses by filtering charged molecular fragments according to their trajectories \cite{massspec}. Moreover, seismograph arrays and the LIGO experiment respectively infer the propagation of vibrations and gravitational waves to triangulate distant rare events \cite{seismo,ligo}.

Given that quantum sensors have previously increased sensitivities for measurements of forces \cite{force} and electric/magnetic fields \cite{electricf,magnetic}, it would be natural to expect that trajectory sensors may also benefit from the use of quantum resources. For instance, a quantum version of the bubble chamber might replace the sensitive medium of a superheated liquid with an array of qubits. Instead of leaving bubbles, an incident particle would interact with the array by applying a local unitary operation to each qubit that it intercepts along its path, and the possible trajectories could ideally be distinguished with a single projective measurement. 
%
%
%
%
%
Prior work involving spatially distributed quantum systems offers promising possibilities for this scenario. Arrays of atoms and superconducting qubits have been shown to be sensitive to incident particles \cite{quantum_particles,Ito2024,mcdermott,oliver}.  
%
%
Furthermore, entangled states have been beneficial for the problem of quantum channel discrimination, which is closely related to trajectory sensing \cite{channel_discrim,Zhuang2020}. Entangled states have also been useful for localizing few-qubit perturbations in a quantum sensor network \cite{guptaPRA,PRAhillery}, even in the presence of realistic noise \cite{gupta}. These successes offer hope that entanglement could similarly enhance quantum sensor arrays aiming to unambiguously distinguish particle trajectories involving many qubits.

The performance of a quantum trajectory sensor could be quantified by the probability of failing to determine the correct particle trajectory after a single measurement, which decreases with the interaction strength $\theta$ between the particle and sensor qubits. Our main question inquires whether entangled sensor states might exist which reduce the $\theta$ required for trajectory sensing to succeed with low failure probability using a single-shot measurement.

The concept of using projective measurements to detect various perturbations on an entangled quantum state is familiar elsewhere as the setting for quantum error correction (QEC), where the goal is instead to preserve the state. By reimagining codes as sensors rather than a means to protect information, our question equivalently asks whether quantum codes exist which allow the ``errors'' imposed by different particle trajectories to be distinguished using a syndrome measurement.

In this work and a companion paper \cite{chin}, we formulate a version of the quantum trajectory sensing (TS) problem where trajectories generally affect many qubits, the set of allowed trajectories is discrete, and the particle-qubit interaction is short-range. Namely, each qubit exactly coincident with the particle's path is rotated by the same angle $\theta\in[0,\pi]$ around some fixed axis of the Bloch sphere, where $\theta$ parameterizes the particle-qubit interaction strength. We solve this problem and show for some threshold $\theta_{\rm min}$ that for all $\theta\in[\theta_{\rm min}, \pi)$, there exist entangled sensor states that perfectly discriminate trajectories in a single shot while any unentangled sensor must instead fail with nonzero probability. For $\theta\in[\theta_{\rm min},\pi)$, an unentangled sensor requires $\Theta(\log(1/\epsilon))$ separate particles to repeatedly pass through the array along the same path to estimate their common trajectory with $\epsilon$ error probability, while an entangled sensor could perfectly determine the trajectory with just one particle. 

Interestingly, these entangled sensor states also improve TS when the interaction strength for each qubit varies continuously with its distance to the particle path. For the problem of tracking a Gaussian-profile laser pulse propagating through an array of atoms, we show that entangled sensors reduce the failure probability by an amount proportional to the beam amplitude in the wide- and weak-beam limit. 

We begin by developing intuition through a minimal example. We then derive $\theta_{\rm min}$ for two TS scenarios and characterize the TS enhancement possible with entanglement. Lastly, we link the problems of TS and error correction and provide a TS scheme which is in principle resilient to noise. This paper focuses on entanglement and the quantum enhancement offered over classical TS, while the companion paper develops the mathematical framework.


\textit{A minimal example.---}Let $\ket{\pm} = (\ket{0}\pm \ket{1})/\sqrt{2}$, and consider a two-qubit sensor $\ket{\psi}$ where the incoming particle rotates one of the qubits by $R_Z(\theta) = e^{-i\theta Z/2}$. If $\theta=\pi$, then the post-trajectory output states of $\ket{\psi}=\ket{++}$ (i.e., $\ket{+-}$ and $\ket{-+}$) are orthogonal and therefore perfectly distinguishable by a single projective measurement.  In contrast, if $\theta<\pi$ (a ``weak" interaction regime), the outputs of this unentangled sensor are not orthogonal and can only be distinguished with some nonzero probability of failure. On the other hand, for $\theta=\frac{\pi}{2}$, the entangled Bell state $\ket{\psi}=\frac{1}{\sqrt{2}}(\ket{01}+\ket{10})$ in fact allows both trajectories to be distinguished with no error, something which would be impossible in the corresponding classical scenario.  Entangled states play a major role in TS and can enable trajectory discrimination even when the interaction strength is weak. 

\textit{Trajectory sensing problem.---}The TS problem is formalized as follows. 
Suppose we are given an array of $n$ qubits labeled $1$ to $n$. A \textit{trajectory} is defined to be a set of qubit indices, and the particle will interact with the array by rotating all of the qubits in its trajectory by $R_Z(\theta)$ for some fixed, precisely known interaction strength $\theta$. Given an $n$-qubit sensor state $\ket{\psi}$ and allowed trajectory set $\mathcal{T}$, each trajectory $T\in\mathcal{T}$ yields a distinct output state $R^{(T)}(\theta)\ket{\psi}$, where $R^{(T)}(\theta)$ applies $R_Z(\theta)$ to each qubit in $T$. A \textit{TS problem} asks for what $\theta$ there exists a \textit{TS state} $\ket{\psi}$ such that all outputs $R^{(T)}(\theta)\ket{\psi}$ for $T\in\mathcal{T}$ are mutually orthogonal and therefore perfectly distinguishable---with zero probability of failure---via a single projective measurement. Given $\theta$, the criteria for $\ket{\psi}$ to be a TS state are summarized by the following orthogonality conditions: 
\begin{align}\label{eq:orth-rel}
    \bra{\psi}R^{\dagger(T)}(\theta)R^{(T')}(\theta)\ket{\psi} = \delta_{T,T'}
\end{align}
for all $T,T'\in\mathcal{T}$. Note that a TS problem is fully determined by the values of $n$ and $\mathcal{T}$. 

\textit{Symmetric trajectory sensors.---}We now introduce a general family of TS problems and provide a working example to illustrate how to systematically determine the range of $\theta$ for which a TS state exists. A TS state is called \textit{symmetric} if the corresponding $\mathcal{T}$ includes all possible trajectories encompassing $m$ qubits for some chosen integer parameter $m\leq n$ (see Figure \ref{fig:examples}a); denote this particular $\mathcal{T}$ with $\Tsym$.
\begin{figure}[htbp]
  \includegraphics[width = \linewidth]{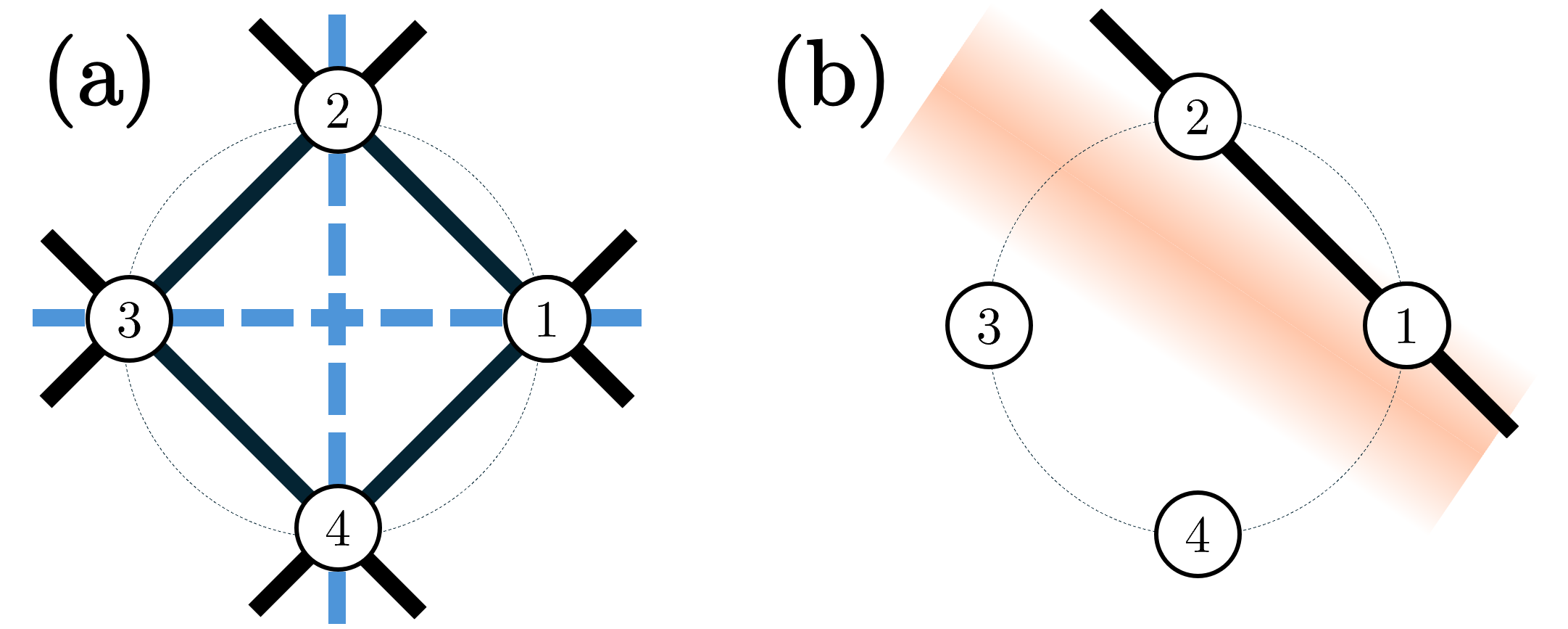}
  \caption{\textbf{(a)} Allowed particle trajectory sets $\Tsym$ (solid and dashed lines) and $\Tcyc$ (solid lines only) when $n=4$ and $m=2$. \textbf{(b)} Gaussian-profile laser pulse and its nearest discrete trajectory in $\Tcyc$.} 
  \label{fig:examples}
\end{figure}

The existence of a TS state for $\Tsym$ and a particular $\theta$ can be determined by substituting the naive general ansatz $\ket{\psi} = \sum_j c_j\ket{j}$ (with $c_j\in\mathbb{C}$ and $Z$-eigenbasis states $\ket{j}$) into Eq. (\ref{eq:orth-rel}) for all $T,T'\in\Tsym$ and checking if the resulting system admits a solution of $c_j$. A trivial solution is evident at $\theta=\pi$ (i.e., $\ket{\psi} = \ket{+}^{\otimes n}$), so we seek a nontrivial solution at any $\theta < \pi$. For general $n$ and $m$, however, this system includes $\abs{\Tsym}^2 = \binom{n}{m}^2$ equations in $2^n$ complex variables, which is computationally intractable for large $n$.

Two symmetries of the symmetric TS problem can be leveraged to greatly simplify this system. The first is a permutation symmetry: $\mathcal{T}_{\rm sym}$ does not change after permuting the index labels assigned to each qubit. Consequently, the desired TS state continues to satisfy Eq. (\ref{eq:orth-rel}) for $\mathcal{T}_{\rm sym}$ under any permutation of the indices. 

A bit-flip symmetry also holds. Note that the complex conjugation of Eq. (\ref{eq:orth-rel}) gives
\begin{align}
    \delta_{T,T'} &= \left(\bra{\psi}R^{\dagger(T)}R^{(T')}\ket{\psi}\right)^*= \bra{\psi}\left(R^{\dagger(T)}R^{(T')}\right)^\dagger \ket{\psi}\nonumber\\
    &= \bra{\psi}X^{\otimes n}\left(R^{\dagger(T)}R^{(T')}\right)X^{\otimes n} \ket{\psi}\label{eq:bit-flip}
\end{align}
for any $T,T'\in\Tsym$, since each $R^{(T)}(\theta)$ operator is a tensor product of single-qubit $R_Z(\theta)$ and identity operators and $XR_ZX = R_Z^\dagger$. Eq. (\ref{eq:bit-flip}) implies that $X^{\otimes n}\ket{\psi}$ also satisfies Eq. (\ref{eq:orth-rel}) if $\ket{\psi}$ does.

When $\mathcal{T}=\Tsym$, the fact that Eq. (\ref{eq:orth-rel}) remains satisfied under the action of these two symmetries means that the search for TS states can be restricted to a much smaller symmetrized subspace; specifically, a TS state exists for $\Tsym$ if and only if there exists a TS state which is permutation and bit-flip invariant \cite{chin}. Thus, a simplified TS state ansatz can be constructed using the following basis states which span this invariant space:
\begin{align}
    \ket{\overline{\nu}} = \sum_{j_1\ldots j_n\in\mathcal{W}_\nu\cup\mathcal{W}_{n-\nu}}\ket{j_1\ldots j_n}
\end{align}
for $\nu=0,\ldots, \floor*{\frac{n}{2}}$, where the $\ket{j_1\ldots j_n}$ are $Z$-eigenbasis states and $\mathcal{W}_{\nu}$ is the set of all length-$n$ bit-strings $j_1\ldots j_n$ with weight $\nu$ (i.e., such that $\sum_k j_k = \nu$).
These states can equivalently be viewed as superpositions of Dicke states \cite{dicke}. 

The specific case of $\mathcal{T}=\mathcal{T}_{\rm sym}(n=4,m=2)$ illustrates how to cleanly determine bounds on the interval of achievable $\theta$ using these symmetry-based simplifications. Substituting the symmetrized TS state ansatz $\ket{\psi}=\sum_{\nu=0}^{n/2}\overline{c}_\nu\ket{\overline{\nu}}$ into Eq. (\ref{eq:orth-rel}) gives a system in just 3 variables $\overline{c}_\nu\in\mathbb{C}$ as opposed to the $2^4$ variables used in the naive approach, and all but 3 of the equations become redundant. Accordingly, the existence of a TS state for a given value of $\theta$ is determined by the solution to a 3-by-3 linear system:
\begin{align}
    \begin{cases}
        1&=\abs{\overline{c}_0}^2 + 4\abs{\overline{c}_1}^2 +3\abs{\overline{c}_2}^2\\
        0&=\abs{\overline{c}_0}^2 + 2\abs{\overline{c}_1}^2(1+\cos{\theta}) + \abs{\overline{c}_2}^2 (1+2\cos{\theta})\\
        0&=\abs{\overline{c}_0}^2 + 4\abs{\overline{c}_1}^2\cos{\theta} + \abs{\overline{c}_2}^2 (2+\cos{2\theta}).
    \end{cases}
\end{align}
Assuming $\theta\neq 0$, this system transforms into a normalization condition along with two constraints $\abs{\overline{c}_0}^2 = \cos(2\theta)\abs{\overline{c}_2}^2$ and $\abs{\overline{c}_1}^2=-\cos(\theta)\abs{\overline{c}_2}^2$. Because the $\abs{\overline{c}_\nu}^2$ must be nonnegative, a solution requires $\cos{2\theta}\geq 0$ and $\cos{\theta}\leq 0$. Subsequently, a TS state exists if and only if $\theta\geq \frac{3\pi}{4}$, proving that TS states can be found in the nontrivial weak-interaction regime.

For general even $n$ with $m=n/2$, the analogous symmetrized system has a solution if the $\overline{c}_\nu$ similarly obey 
\begin{align}
    \abs{\overline{c}_\nu}^2 = (-1)^{m-k}\cos\left[\left(m-\nu\right)\theta\right]\abs{\overline{c}_{m}}^2
\end{align}
for all $\nu=0,\ldots,n/2$, and the requirement that $\abs{\overline{c}_\nu}^2 \geq 0$ bounds the achievable $\theta$ per the following theorem (proved fully in \cite{chin}):
\begin{theorem}\label{thm:S-TS}
    Suppose $\mathcal{T}=\Tsym$. For $\theta\in[0,\pi]$ and arbitrary $n>0$ and $m\geq 0$, a sufficient criterion for the existence of a TS state is 
    \begin{align}\label{eq:S-TS-criterion}
        \theta \geq \frac{(n-1)\pi}{n}.
    \end{align}
    Furthermore, when $m=\floor*{\frac{n}{2}}$ or $\ceil*{\frac{n}{2}}$, Eq. (\ref{eq:S-TS-criterion}) becomes a necessary criterion.
\end{theorem}

Theorem \ref{thm:S-TS} confirms that TS states exist for arbitrarily sized 
$\mathcal{T}_{\rm sym}$ at some nontrivial $\theta<\pi$. However, the minimum $\theta$ needed for a TS state increases towards $\pi$ as the number of qubits $n$ grows. This loss of ``sensitivity'' is intuitively follows from the fact that the number $\abs{\Tsym} = \binom{n}{m}$ of trajectories to be distinguished generally scales much faster than the sensor size $n$. This observation suggests that the sensitivity of a TS state to the interaction strength $\theta$ might be increased by restricting the trajectories in $\mathcal{T}$.

\textit{Cyclic trajectory sensors.---}Since particles like neutrinos and dark matter likely interact with lower $\theta$ than is achievable for $\mathcal{T}_{\rm sym}$, it would be desirable to find alternative TS states for even smaller values of $\theta$. Note that $\mathcal{T}_{\rm sym}$ includes many trajectories which may be unphysical in a practical setting (e.g., where the constituent qubits are not localized together along a continuous curve). In fact, many experimental applications may require relatively few trajectories---for example, neutrino paths might only comprise straight lines. 

For these reasons, we now introduce \textit{cyclic} TS states, where $\mathcal{T}$ is restricted to include only ``continuous" trajectories where the $m\leq n$ constituent qubits have consecutive indices modulo $n$ (see Figure \ref{fig:examples}a); denote such a $\mathcal{T}$ with $\Tcyc = \{z^j(\{1,\ldots,m\})\ :\ j=1,\ldots, n\}$, where $z = (1\ldots n)$ is the cyclic permutation of $n$ indices. Observe that $\abs{\Tcyc} = n$ as opposed to $\binom{n}{m}$ in the symmetric case. As before, we ask for what $\theta$ there exists a TS state satisfying Eq. (\ref{eq:orth-rel}) for all $T,T'\in \Tcyc$; however, the naive approach using a completely general TS ansatz still remains computationally intractable for large $n$.

The key insight is that there exist TS states for $\Tcyc$ which can be decomposed as the tensor product of multiple identical, smaller TS states, and this observation will allow for the system of Eq. (\ref{eq:orth-rel}) to be greatly simplified. The $\mathcal{T}_{\rm cyc}(n=4,m=2)$ example illustrates this simplification. For a given $\theta$, define $\ket{\phi}$ to be a two-qubit TS state for the smaller $\mathcal{T}_{\rm sym}(n'=2,m'=1)$ problem. Preparing each of the qubit pairs $\{1,3\}$ and $\{2,4\}$ into $\ket{\phi}$, the resulting 4-qubit state $\ket{\psi}=\ket{\phi}_{1,3}\otimes \ket{\phi}_{2,4}$ is in fact a TS state for $\mathcal{T}_{\rm cyc}(n=4,m=2)$ at the same $\theta$. 

This assertion is justified by considering the action of two trajectories $T=\{1,2\}$ and $T'=\{2,3\}$ on $\ket{\psi}$. To show the left side of Eq. (\ref{eq:orth-rel}) equals zero, the expression can be factored as 
\begin{align}
    &\bra{\psi}R^{\dagger(T)}R^{(T')}\ket{\psi}\nonumber\\
    &= \bra{\psi}(R^\dagger_Z\otimes R^\dagger_Z \otimes I\otimes I)(I\otimes R_Z \otimes R_Z\otimes I)\ket{\psi}\nonumber\\
    &= \left(\bra{\phi}_{1,3}(R^{\dagger}_Z\otimes R_Z)\ket{\phi}_{1,3}\right)\left(\bra{\phi}_{2,4}(I\otimes I)\ket{\phi}_{2,4}\right)=0
\end{align}
due to the fact that $(R_Z\otimes I)\ket{\phi}$ and $(I\otimes R_Z)\ket{\phi}$ are orthogonal by the definition of $\ket{\phi}$. A similar argument holds for any other choice of $T,T'\in\Tcyc$, from which we conclude that $\ket{\psi}$ is indeed the desired TS state. Moreover, since this $\ket{\psi}$ exists if a suitable $\ket{\phi}$ exists, the desired range of achievable $\theta$ includes any $\theta$ feasible for the smaller $\mathcal{T}_{\rm sym}(n'=2,m'=1)$ problem. Thus, by Theorem \ref{thm:S-TS}, a TS state $\ket{\psi}$ for $\mathcal{T}_{\rm cyc}(n=4,m=2)$ exists if $\theta\geq\frac{\pi}{2}$. Limiting the set of allowed trajectories from $\mathcal{T}_{\rm sym}$ to $\mathcal{T}_{\rm cyc}$ therefore expands the range of achievable $\theta$ from $\left[\frac{3\pi}{4},\pi\right]$ to $\left[\frac{\pi}{2},\pi\right]$; intuitively, decreasing the number of trajectories lowers the minimum particle-qubit interaction strength needed to distinguish them. 

For $\Tcyc$, this method of constructing TS states generalizes to show that, so long as $n=\kappa m$ for some positive integer $k>1$, a TS state exists which decomposes into $m$ identical copies of a smaller TS state for $\mathcal{T}_{\rm sym}(n'=\kappa,m'=1)$ \cite{chin}. Since the larger cyclic TS state exists over the range of $\theta$ for which the smaller symmetric TS state exists, the range of achievable $\theta$ can be expressed entirely in terms of the size parameter $\kappa$ of the smaller symmetric problem. Solving Eq. (\ref{eq:orth-rel}) for this smaller problem leads to the following theorem (also proved fully in \cite{chin}):
\begin{theorem}\label{thm:C-TS}
    Suppose $\mathcal{T}=\Tcyc$, where $n=\kappa m$ for some arbitrary $m>0$ and $\kappa>1$. Given $\theta\in[0,\pi]$, a sufficient criterion for the existence of a TS state is
    \begin{align}\label{eq:C-TS-criterion}
        \theta \geq \arccos{\left(-1+\ceil*{\frac{\kappa}{2}}^{-1}\right)}.
    \end{align}
\end{theorem}
In contrast to Theorem \ref{thm:S-TS}, Theorem \ref{thm:C-TS} shows for $\Tcyc$ that, so long as $n$ is a constant multiple of $m$, the minimum $\theta$ needed for a TS state remains constant as the sensor and trajectory set grow. For example, when $n=2m$, a symmetric TS state requires $\theta\approx 0.95\pi$ to distinguish about a million trajectories, whereas a cyclic TS state only requires $\theta=0.5\pi$! Note that recent work \cite{gupta} has independently proved a separate result equivalent to Theorem \ref{thm:C-TS} for the special case of $m=1$.

\textit{Quantum enhancement from entanglement.---}The powerful ability of symmetric and cyclic TS states to perfectly distinguish trajectories in one shot for weak $\theta<\pi$ is a direct consequence of entanglement. Intuitively, entanglement allows the $\theta$-rotations applied to each qubit in a trajectory to add constructively such that Eq. (\ref{eq:orth-rel}) can be satisfied. The following theorem rigorously asserts that there are no unentangled TS states if $\theta<\pi$:
\begin{theorem}\label{thm:entanglement}
    For arbitrary $n>0$ and $\mathcal{T}$ with $\abs{\mathcal{T}} > 1$, a fully unentangled TS state of the form $\bigotimes_{i=0}^n\ket{\psi_i}$ for some single-qubit states $\ket{\psi_i}$ exists if and only if $\theta = \pi$.
\end{theorem}
\begin{proof}
    ``$\implies$'' direction: The expression obtained by substituting $\bigotimes_{i=0}^n\ket{\psi_i}$ into the left side of Eq. (\ref{eq:orth-rel}) can be expanded as a product of terms that look like $\bra{\psi_i}A_i\ket{\psi_i}$, where $A_i=I,R_Z(\theta),$ or $R^\dagger_Z(\theta)$. Eq. (\ref{eq:orth-rel}) for any $T\neq T'$ hence implies that some $\bra{\psi_i}R_Z(\theta)\ket{\psi_i} = 0$, which requires that $\theta=\pi$. ``$\impliedby$'' direction: if $\theta=\pi$, then $\ket{+}^{\otimes n}$ is a TS state.
\end{proof}
It follows that symmetric and cyclic TS states must be entangled if $\theta<\pi$. Hence, entangled TS states can perfectly distinguish trajectories at lower $\theta$ than unentangled TS states can.

For the problem of distinguishing a discrete trajectory set $\mathcal{T}$, this quantum enhancement can be appreciated by comparing, as a function of $\theta$, the single-shot failure probabilities of a quantum protocol utilizing an entangled TS state against a ``classical" protocol which instead uses an unentangled state as a sensor (details in \cite{supplement}).  For $\mathcal{T}=\mathcal{T}_{\rm sym}(n=4,m=2)$, the quantum protocol succeeds with lower failure probability for all nontrivial values $\theta\in(0,\pi)$ of the interaction strength (Figure \ref{fig:quant-adv}). Furthermore, the quantum protocol succeeds with \textit{zero} failure probability for all $\theta\geq\frac{3\pi}{4}$ satisfying Eq. (\ref{eq:S-TS-criterion}) (highlighted region), since Theorem \ref{thm:S-TS} guarantees the existence of a TS state over this interval of $\theta$. In contrast, so long as $\theta<\pi$, the classical protocol fails with some nonzero probability because Theorem \ref{thm:entanglement} guarantees that no unentangled TS state exists over this range of $\theta$. More generally, for arbitrary $n$ and $m$, entanglement-enhanced trajectory sensing can succeed perfectly in one shot over an interval of $\theta\in[\theta_{\rm min},\pi)$ where a protocol without entanglement cannot; $\theta_{\rm min}$ is the bound given by Eq. (\ref{eq:S-TS-criterion}) or (\ref{eq:C-TS-criterion}) when $\mathcal{T}=\mathcal{T}_{\rm sym}$ or $\mathcal{T}_{\rm cyc}$, respectively. 
\begin{figure}[htbp]
\vspace*{-2ex}
  \includegraphics[width = \linewidth]{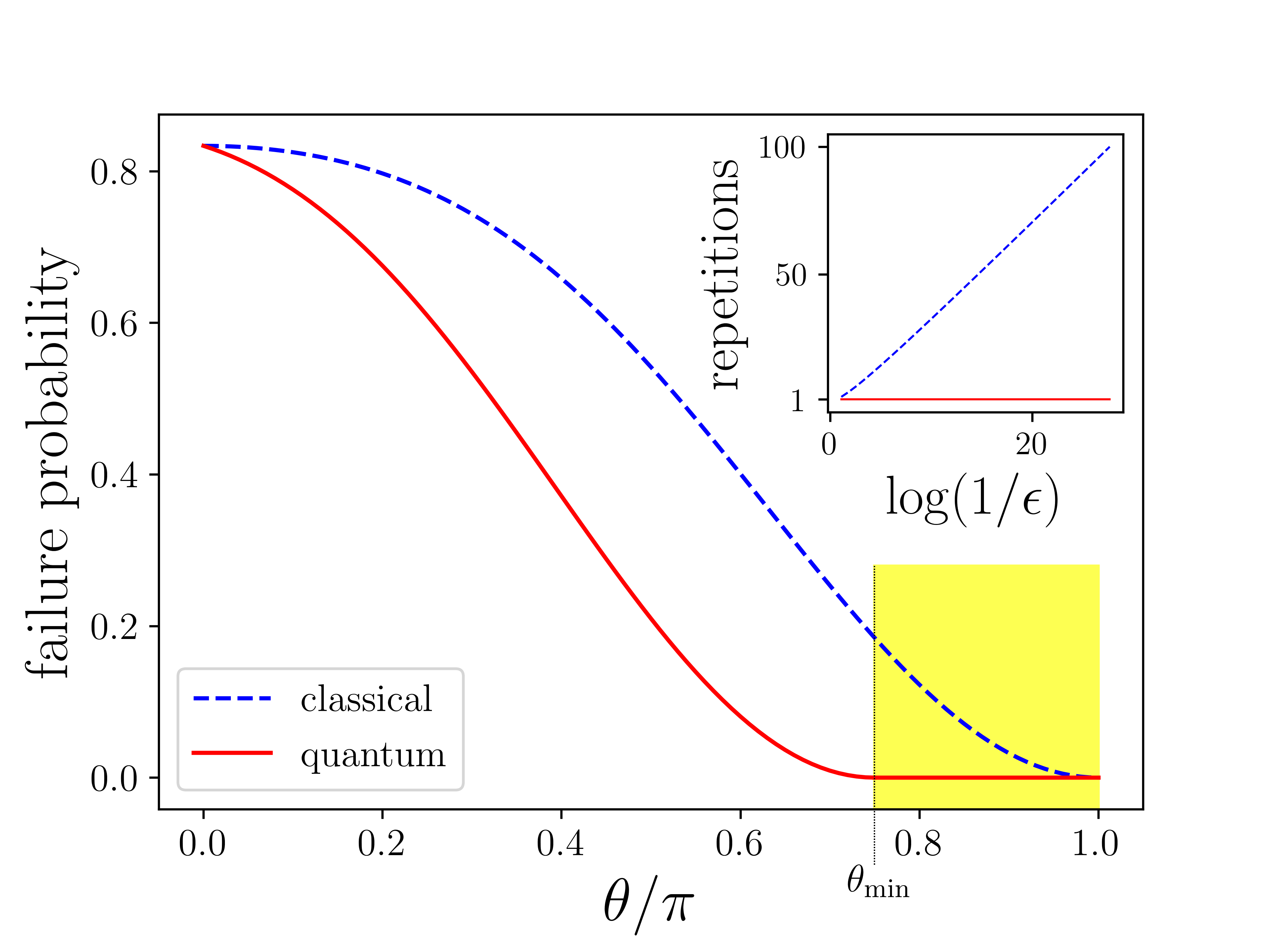}
  \caption{Single-shot failure probabilities of classical and quantum TS protocols vs. particle-qubit interaction strength $\theta$ when $\mathcal{T}=\mathcal{T}_{\rm sym}(n=4,m=2)$. Inset shows number of TS protocol repetitions required to estimate trajectory with error probability $\epsilon$ when $\theta=\theta_{\rm min}= \frac{3\pi}{4}$.} 
  \label{fig:quant-adv}
\end{figure}

Entangled sensors can exhibit decisive benefits even if the TS protocol allows for multiple repeated measurements. Suppose $r$ particles sequentially pass through the array along the \textit{same} trajectory $T$, and we measure and reinitialize the TS state between each particle interaction. We then perform a majority vote on the $r$ measurement outcomes to estimate $T$. The inset plot to Figure \ref{fig:quant-adv} shows for the protocols in the above example how large $r$ must be for the majority vote to propose $T$ with $\epsilon$ error probability when $\theta=\theta_{\rm min} =\frac{3\pi}{4}$. For an unentangled sensor with $\theta\in[\theta_{\rm min},\pi)$, the probability that the majority vote fails is given by the lower tail of a binomial distribution in $r$ trials, whose size is well-known to decrease exponentially with $r$. Subsequently, $r$ must be on the order of $\Theta(\log (1/\epsilon))$ for the majority vote to succeed with error probability $\epsilon$. On the other hand, an entangled sensor has perfect one-shot success probability over this range of $\theta$ and would only require a single particle to determine $T$ with zero error.

\textit{TS with distance-dependent interactions.}---In an experimental setting, it may be desirable to characterize a continuum of possible particle paths, especially when the particle has a distance-dependent interaction with the sensor qubits. For example, suppose a Gaussian-profile laser pulse of an arbitrary, unknown direction is incident upon an array of four atoms (numbered 1-4) such that an internal qubit of atom $i$ is rotated by an angle $\theta_i = \theta_{0}\exp(-\frac{d_i^2}{w^2})$, where $d_i$ is the distance from atom $i$ to the center of the beam, $\theta_{0}$ is proportional to the beam amplitude times the pulse duration, and $w$ is the beam waist (see Figure \ref{fig:examples}b). By preparing the four qubits into a TS state, it is possible to use a projective measurement to estimate the trajectory in some discrete set $\mathcal{T}$ which aligns best with the unknown beam path (details in \cite{supplement}). Let $\mathcal{T}=\mathcal{T}_{\rm cyc}(n=4,m=2)$, and assume $\theta_{0}\ll \pi $ and $w \gg 1$. If an entangled TS state is used as the sensor, then the one-shot failure probability for estimating the nearest discrete trajectory (averaged over possible beam paths) is $\frac{3}{4}-\frac{8}{\pi^2w^2}\theta_{0}$ to first order in $\theta_0$ and $\frac{1}{w^2}$. In contrast, the average failure probability for an unentangled sensor in the same limit is $\frac{3}{4}$. It follows that the TS enhancement due to entanglement increases linearly with the maximum beam amplitude when the beam is weak and wide.

\textit{Connection to error correction.---}A beautiful connection emerges between TS and quantum error correction when the trajectories are instead interpreted as errors to be identified.  Given $n$ and a trajectory set $\mathcal{T}$, consider the quantum ``error" channel 
\begin{align}\label{eq:errors}
    \mathcal{E}(\rho)=\frac{1}{\abs{\mathcal{T}}}\sum_{T\in\mathcal{T}}R^{(T)}(\theta)\rho R^{\dagger(T)}(\theta),
\end{align}
which corresponds to the action of picking an unknown trajectory $T\in\mathcal{T}$ uniformly at random and applying it to the input state density matrix $\rho$. In fact, any (pure) TS state $\rho_{\rm TS} = \ketbra{\psi_{\rm TS}}$ is actually a QEC code state for this error channel, since a single syndrome measurement can precisely reveal which error $R^{(T)}$ was applied to $\rho_{\rm TS}$, and $\rho_{\rm TS}$ can then be recovered by applying the operator $R^{\dagger(T)}$. Furthermore, the TS existence criteria in Eq. (\ref{eq:orth-rel}) correspond to a special case of the Knill-Laflamme QEC criteria \cite{knill} for $\mathcal{E}$ and the one-dimensional code $\{\ket{\psi}\}$. 

Although many conceivable codes might \textit{recover} the errors from $\mathcal{E}$ in Eq. (\ref{eq:errors}), TS codes are distinguished by the special ability to perfectly \textit{discriminate} all possible errors. 
Moreover, while QEC typically focuses on independent and identically-distributed single-qubit Pauli errors, trajectories instead involve highly correlated non-Pauli operations affecting many qubits.

It is natural to inquire whether the existing vast literature on quantum codes can be used to solve TS problems. Indeed, the cyclic TS state $\frac{1}{2}(\ket{0011}+\ket{0110}+\ket{1100}+\ket{1001})$ for $\mathcal{T}_{\rm cyc}(n=4,m=2)$ and $\theta = \frac{\pi}{2}$ spans a stabilizer code with generators $\langle -Z_1Z_3,-Z_2Z_4,X_1X_3,X_2X_4\rangle$; this code is a subcode of the $[[4,2,2]]$ CSS code \cite{422code}, which is the smallest toric code \cite{toric}. In fact, larger instances of the toric code can also support TS states \cite{chin}. However, symmetric TS states for $\mathcal{T}_{\rm sym}(n=4,m=2)$ with $\theta<\pi$ instead constitute a new variety of non-stabilizer permutation-invariant codes, suggesting that such codes might be useful for discriminating correlated non-Pauli errors. 

Furthermore, by concatenating TS states with familiar quantum codes, perfect one-shot trajectory sensing may in principle be achieved even with decoherence. Note that the $[[7,1,3]]$ Steane code \cite{steane} transversally implements the $R_Z\left(\frac{\pi}{2}\right)$ gate, i.e., applying $R_Z\left(\frac{\pi}{2}\right)$ to each physical qubit implements $R^\dagger_Z\left(\frac{\pi}{2}\right)$ on the logical qubit \cite{shor1997faulttolerantquantumcomputation}. Thus, if the logical qubits of multiple Steane code blocks are prepared to some symmetric/cyclic TS state with parameters $n,m$, and $\theta = \frac{\pi}{2}$, then the physical qubits together constitute a larger TS state with parameters $7n, 7m$, and $\theta=\frac{\pi}{2}$, where the trajectories consist of whole blocks. Because the Steane code corrects arbitrary single-qubit errors, this concatenated TS state succeeds perfectly even if a correctable error occurs in any of the blocks during the protocol.

\textit{Conclusion.---}Single-shot trajectory sensing using entangled TS states is particularly promising for experimental applications where the particles rarely interact with the sensor. Realistically, very weakly interacting particles such as dark matter may require a $\theta$ smaller than considered here. However, for many applications, $\mathcal{T}_{\rm sym}$ and $\mathcal{T}_{\rm cyc}$ may include more trajectories than needed, and further restricting $\mathcal{T}$ should continue to decrease the achievable $\theta$. Note also that our sensors determine a particle's trajectory assuming it has already positively entered the device; as such, it would be desirable to augment these sensors with another for simple detection. Additionally, although TS states may be susceptible to decoherence due to their entanglement, they are fortunately typically not maximally entangled and in fact must utilize unusual multipartite entanglement. Notably, the Dicke states \cite{dicke} used to construct symmetric TS states are known to be relatively robust against decoherence \cite{dicke-decohere} and have interesting entanglement properties \cite{dicke-entanglement,toth,dicke-entanglement-evolution}. Finally, given that TS states are special quantum code states, we ask how known families of quantum codes might enable new trajectory sensing capabilities. 

\begin{acknowledgments}
 Z.\,E.\,C. acknowledges support from the National Science Foundation Graduate Research Fellowship under Grant No. 2141064. D.\,R.\,L. acknowledges support from the NSF Q-SAIL National Quantum Virtual Laboratory under Grant No. 2410687. I.\,L.\,C. acknowledges support by the NSF Center for Ultracold Atoms. 
\end{acknowledgments}

\bibliography{References}
\bibliographystyle{apsrev4-1}
\end{document}


\title{Supplemental Material for: Quantum entanglement enables single-shot trajectory sensing for weakly interacting particles}
\author{Zachary E. Chin} \email{zchin@mit.edu}\affiliation{Department of Physics, Massachusetts Institute of Technology, Cambridge, Massachusetts 02139, USA}
\author{David R. Leibrandt}\affiliation{Department of Physics, University of California, Los Angeles, CA 90095, USA}
\author{Isaac L. Chuang}\affiliation{Department of Physics, Center for Ultracold Atoms, and Research Laboratory of Electronics, \\ Massachusetts Institute of Technology, Cambridge, Massachusetts 02139, USA} 

\maketitle
\onecolumngrid

This Supplemental Material describes the classical and quantum trajectory sensing (TS) protocols used in the main text. First, we explain the protocols for discriminating a discrete trajectory set whose failure probabilities are plotted in Figure 2. Afterwards, we discuss how TS states can be used for estimating the path of a Gaussian-profile laser pulse through an array of atoms, and we derive the corresponding failure probabilities for entangled and unentangled sensors. 

In the following discussion, let the notation $\ket{\psi_{\theta}}$ represent a TS state satisfying the orthogonality conditions for the given $\mathcal{T}$ at an arbitrary value of $\theta$, assuming it exists. Additionally, let $\theta_{\rm min}$ be the smallest interaction strength such that the orthogonality conditions admit a solution for all $\theta\in[\theta_{\rm min},\pi]$.

\textit{Classical protocol for discrete TS.---}The following classical TS protocol does not use entangled sensors and serves as a performance benchmark in Figure 2. The protocol takes a trajectory set $\mathcal{T}$ and interaction strength $\theta\in[0,\pi]$ as known inputs and proceeds as follows. First, the unentangled sensor state $\ket{\psi_\pi}=\ket{+}^{\otimes n}$ is prepared on the qubit array. Note that $\ket{+}^{\otimes n}$ is an optimal unentangled sensor state since $\ket{+}$ minimizes the inner product $\bra{+}R_Z(\theta)\ket{+}$ for all $\theta$. Then, an unknown trajectory $T$ is chosen uniformly at random from $\mathcal{T}$ and transforms the sensor state to $R^{(T)}(\theta)\ket{\psi_\pi}$. Each qubit of this output state is subsequently measured in the $\{\ket{+},\ket{-}\}$ basis; let $S$ be the set of qubits which were measured to be $\ket{-}$. The proposed estimate for $T$ is selected uniformly at random from the set $\{T'\in\mathcal{T}\ :\ S\subseteq T'\}$. 

\textit{Quantum protocol for discrete TS.---}The following quantum trajectory sensing protocol is used in Figure 2 to illustrate the TS enhancement possible with entangled sensors. As above, the protocol takes $\mathcal{T}$ and $\theta\in[0,\pi]$ as known inputs. Note that when $\mathcal{T}=\mathcal{T}_{\rm sym}(n,m)$ as in Figure 2, then $\theta_{\rm min}= \frac{(n-1)\pi}{n}$ is the bound given by Theorem 1. First suppose that $\theta \geq \theta_{\rm min}$. Then the TS state $\ket{\psi_\theta}$ is prepared and an unknown, uniformly random trajectory $T\in\mathcal{T}$ takes the state to $R^{(T)}(\theta)\ket{\psi_\theta}$. The output state is measured in the orthonormal (partial) basis of possible outputs $\{R^{(T')}(\theta)\ket{\psi_\theta}\ :\ T'\in\mathcal{T}\}$, allowing $T$ to be determined exactly with unit probability. 

Now suppose that $\theta<\theta_{\rm min}$. Then a TS state yielding orthogonal outputs is no longer guaranteed to exist. In this case, $\ket{\psi_{\theta_{\rm min}}}$ is thus always chosen as the input. The output state $R^{(T)}\ket{\psi_{\theta_{\rm min}}}$ resulting from $T$ is then subjected to the projective measurement with projectors
\begin{align}\label{eq:proj-meas}
    \left\{P_{T'}\ \forall T'\in \mathcal{T},\ P^{\perp}\right\},
\end{align}
where 
\begin{align}
    P_{T'}=R^{(T')}(\theta_{\rm min})\ketbra{\psi_{\theta_{\rm min}}}R^{\dagger(T')}(\theta_{\rm min})
\end{align}
for all $T'\in\mathcal{T}$ and 
\begin{align}
    P^{\perp} = I-\sum_{T'\in\mathcal{T}}P_T'.
\end{align}
If the measurement returns some $P_{T'}$, then $T'$ is proposed as the estimate for $T$. Alternately, if the measurement returns $P^{\perp}$, then a trajectory is proposed uniformly at random from $\mathcal{T}$.

\textit{Path estimation for a Gaussian-profile laser pulse.}---Assume there is a sensor array of four identical atoms equally spaced on a unit circle (numbered 1-4) and that the center of an incident Gaussian-profile laser pulse intersects the circle at least once (see Figure 1b). Suppose that each atom has an electronic transition resonant with the laser frequency such that the internal qubit of atom $i$ is rotated by an angle $\theta_i = \theta_{0}\exp(-\frac{d_i^2}{w^2})$, where $d_i$ is the distance from atom $i$ to the center of the beam, $\theta_{0}$ is proportional to the beam amplitude times the pulse duration, and $w$ is the beam waist. We seek to determine the trajectory in a given discrete set $\mathcal{T}$ which is aligns best with the unknown beam path; we assume that all viable beam paths are equally probable. In the following example, let $\mathcal{T}=\mathcal{T}_{\rm cyc}(n=4,m=2)$ so that $\mathcal{T} = \{\{1,2\},\{2,3\},\{3,4\},\{4,1\}\}$.

The interaction of the beam with the qubits depends on the direction of the beam, which can take a continuum of values. Because the atoms are arranged on the unit circle, it will be useful to describe their positions with polar coordinates $(r,\phi)$, where $r\geq 0$ and $\phi\in(-\pi,\pi]$. Thus, suppose atoms 1, 2, 3, and 4 have polar coordinates $(1,0), (1,\frac{\pi}{2}), (1,\pi)$, and $(1,-\frac{\pi}{2})$, respectively, and assume that the beam center intersects the unit circle at points $(1,\phi_1)$ and $(1,\phi_2)$. Since $\phi_1$ and $\phi_2$ fully characterize the beam path, $d_i$ and $\theta_i$ can then be written as functions of $\phi_1$ and $\phi_2$. The beam acts on the state of the four qubits via the operator
\begin{align}
    R_g(\phi_1,\phi_2) = \bigotimes_{i=1}^4 R_Z(\theta_i(\phi_1,\phi_2)),
\end{align}
where the subscript $g$ stands for ``Gaussian".

To estimate the trajectory in $\mathcal{T}$ which is closest to the beam, we will adapt the classical (unentangled) and quantum (entangled) protocols for discrete TS introduced above. The nearest trajectory $T_{\rm min}$ is defined to be the trajectory in $\mathcal{T}$ minimizing the distance
\begin{align}
    \operatorname{dist}(T,\phi_1,\phi_2) = \sum_{i\in T} \abs{d_i(\phi_1,\phi_2)}.
\end{align}
The classical protocol for path estimation prepares the internal qubits to the unentangled $\ket{\psi_{\pi}} = \ket{+}^{\otimes 4}$ and allows $R_g(\phi_1,\phi_2)$ to perturb the sensor state, where $\phi_1$ and $\phi_2$ are chosen uniformly at random from $(-\pi,\pi]$. The output is measured and an estimated trajectory $T\in \mathcal{T}$ is proposed using the same strategy as for discrete TS. The protocol succeeds if and only if $T=T_{\rm min}$.

The quantum protocol for path estimation also proceeds analogously to its discrete TS counterpart. Note that since $\mathcal{T}=\mathcal{T}_{\rm cyc}$, we can choose $\theta_{\rm min} = \frac{\pi}{2}$ by Theorem 2. Thus, the qubits are prepared to the \textit{entangled} state $\ket{\psi_{\frac{\pi}{2}}} = \frac{1}{2}(\ket{0011}+\ket{0110}+\ket{1100}+\ket{1001})$, and a random $R_g(\phi_1,\phi_2)$ perturbs the sensor. Then, the output is measured using the projectors of Eq. (\ref{eq:proj-meas}), and an estimated $T$ is proposed as in the discrete TS case. Again, the protocol succeeds if and only if $T=T_{\rm min}$.  

To derive the failure probabilities of these classical and quantum path estimation protocols, we must determine explicit expressions for the various $d_i(\phi_1,\phi_2)$. The first step is to derive an equation for the line of the beam center in Cartesian coordinates $(x,y)$. Given that the line passes through $(1,\phi_1)$ and $(1,\phi_2)$, its equation is given in point-slope form by
\begin{align}
    y-\sin{\phi_1} &= \frac{\sin{\phi_2}-\sin{\phi_1}}{\cos{\phi_2}-\cos{\phi_1}}(x-\cos{\phi_1})\\
    &= -\cot(\frac{\phi_1+\phi_2}{2})(x-\cos{\phi_1}).
\end{align}
We can rearrange this equation into the form $ax+by+c=0$:
\begin{align}
    0&=\cos(\frac{\phi_1+\phi_2}{2})x +\sin(\frac{\phi_1+\phi_2}{2})y - \left[ \cos(\frac{\phi_1+\phi_2}{2})\cos{\phi_1} + \sin(\frac{\phi_1+\phi_2}{2})\sin{\phi_1}\right]\\
    &=\cos(\frac{\phi_1+\phi_2}{2})x +\sin(\frac{\phi_1+\phi_2}{2})y - \cos(\frac{\phi_1-\phi_2}{2})
\end{align}
Now define the variables $\varphi_1=(\phi_1+\phi_2)/2$ and $\varphi_2=(\phi_1-\phi_2)/2$. We then rewrite the above equation for the line as 
\begin{align}
    \cos{\varphi_1}x+\sin{\varphi_1}y - \cos{\varphi_2} = 0.
\end{align}
The distance from a point $(x_0,y_0)$ to a line $ax+by+c=0$ is given by \cite{line-dist}
\begin{align}
    \frac{\abs{ax_0+by_0+c}}{\sqrt{a^2+b^2}}.
\end{align}
Since atom 1 is located at the point $(1,0)$ in Cartesian coordinates, it follows that
\begin{align}
    d_1(\varphi_1,\varphi_2) = \abs{\cos{\varphi_1} -\cos{\varphi_2}}.
\end{align}
Due to the rotational symmetry of the qubit array, it is easy to see that $d_i(\phi_1,\phi_2) = d_1(\phi_1-q_i,\phi_2-q_i)$, where $q_i$ is the angular coordinate of atom $i$. Thus, $d_i(\varphi_1,\varphi_2)=d_1(\varphi_1-2q_i,\varphi_2)$, which implies that
\begin{align}
    d_2(\varphi_1,\varphi_2) &= \abs{\sin{\varphi_1} -\cos{\varphi_2}},\\
    d_3(\varphi_1,\varphi_2) &= \abs{\cos{\varphi_1} +\cos{\varphi_2}}, \mathrm{and}\\
    d_4(\varphi_1,\varphi_2) &= \abs{\sin{\varphi_1} +\cos{\varphi_2}}.
\end{align}

We now derive one-shot failure probabilities for the classical and quantum path estimation protocols. Consider the quantum protocol first. Success is achieved by either measuring the projector $P_{T_{\rm min}}$ or measuring the projector $P^{\perp}$ and subsequently guessing $T_{\rm min}$ correctly when choosing a trajectory uniformly at random from $\mathcal{T}$. Letting $\ket{\psi_{\rm out}(\varphi_1,\varphi_2)} = R_g(\varphi_1,\varphi_2)\ket{\psi_{\frac{\pi}{2}}}$ be the output state of the sensor and $\rho_{\mathrm{out}}(\varphi_1,\varphi_2) = \ketbra{\psi_{\rm out}}$ be the corresponding density matrix, the success probability is then
\begin{align}
    \Pr[\mathrm{success}](\varphi_1,\varphi_2) &= \tr[P_{T_{\rm min}}\rho_{\mathrm{out}}(\varphi_1,\varphi_2)]+\frac{1}{\abs{\mathcal{T}}}\tr[P^{\perp}\rho_{\mathrm{out}}(\varphi_1,\varphi_2)]   
\end{align}
For clarity, we will use the notation $\ket{\psi_{\frac{\pi}{2}}^{(T)}} = R^{(T)}\left(\frac{\pi}{2}\right)\ket{\psi_{\frac{\pi}{2}}}$ so that $P_T = \ketbra{\psi_{\frac{\pi}{2}}^{(T)}}$. We can then write the success probability as 
\begin{align}
     \Pr[\mathrm{success}](\varphi_1,\varphi_2) &= \abs{\bra{\psi_{\frac{\pi}{2}}^{(T_{\rm min})}}R_g(\varphi_1,\varphi_2)\ket{\psi_{\frac{\pi}{2}}}}^2+\frac{1}{4}\tr[P^{\perp}\rho_{\mathrm{out}}(\varphi_1,\varphi_2)]
\end{align}
since $\abs{\mathcal{T}}=4$. We claim that the second term in the above equation evaluates to zero. To see this, let $V= \operatorname{Span}(\ket{0011},\ket{0110},\ket{1100},\ket{1001})$ and $W=\operatorname{Span}\left(\ket{\psi_{\frac{\pi}{2}}^{(T)}}:T\in\mathcal{T}\right)$. Note that $\dim V = 4$ and that $\dim W = \abs{\mathcal{T}}=4$ since all the $\ket{\psi_{\frac{\pi}{2}}^{(T)}}$ are orthogonal and therefore linearly independent. Furthermore, because $\ket{\psi_{\frac{\pi}{2}}^{(T)}}\in V$ for all $T\in\mathcal{T}$ and $\dim V = \dim W$, we have $V = W$. Consequently, since $\ket{\psi_{\rm out}}\in V$, it is also true that $\ket{\psi_{\rm out}}\in W$. Noting that $\left(\sum_{T\in\mathcal{T}}P_T\right)$ is the projector onto $W$, we thus have $\left(\sum_{T\in\mathcal{T}}P_T\right)\ket{\psi_{\rm out}} = \ket{\psi_{\rm out}} $, so $P^{\perp}\ket{\psi_{\rm out}} =0$. It follows that $\tr[P^{\perp}\rho_{\rm out}]=0$ as well, so the success probability simplifies to
\begin{align}\label{eq:q-succ-exact}
    \Pr[\mathrm{success}](\varphi_1,\varphi_2) &= \abs{\bra{\psi_{\frac{\pi}{2}}^{(T_{\rm min})}}R_g(\varphi_1,\varphi_2)\ket{\psi_{\frac{\pi}{2}}}}^2.
\end{align}

We ultimately seek the success probability averaged over $\varphi_1$ and $\varphi_2$ (i.e., over the equally-probable possible beam paths). However, $T_{\rm min}$ depends on $\varphi_1$ and $\varphi_2$, as can be seen in Figure \ref{fig:integration}. Note that since $\phi_1,\phi_2\in(-\pi,\pi]$, we require that $-\pi<\varphi_1+\varphi_2\leq \pi$ and $-\pi<\varphi_1-\varphi_2\leq \pi$. Due to symmetry, the average success probability over all allowed $\varphi_1,\varphi_2$ is equal to the average probability over any one of the four identically-colored regions in Figure \ref{fig:integration}. To understand this fact, note that the average probabilities over each of the four regions must be the same since the problem remains invariant under successive $\pi/2$ rotations of the array (these rotations realize cyclic permutations of the qubit indices). Consequently, without loss of generality, the average probability can be computed as the average over the (blue) region of $\varphi_1$ and $\varphi_2$ which corresponds to $T_{\rm min} = \{1,2\}$.
\begin{figure}[htbp]
  \includegraphics[width = 0.5\linewidth]{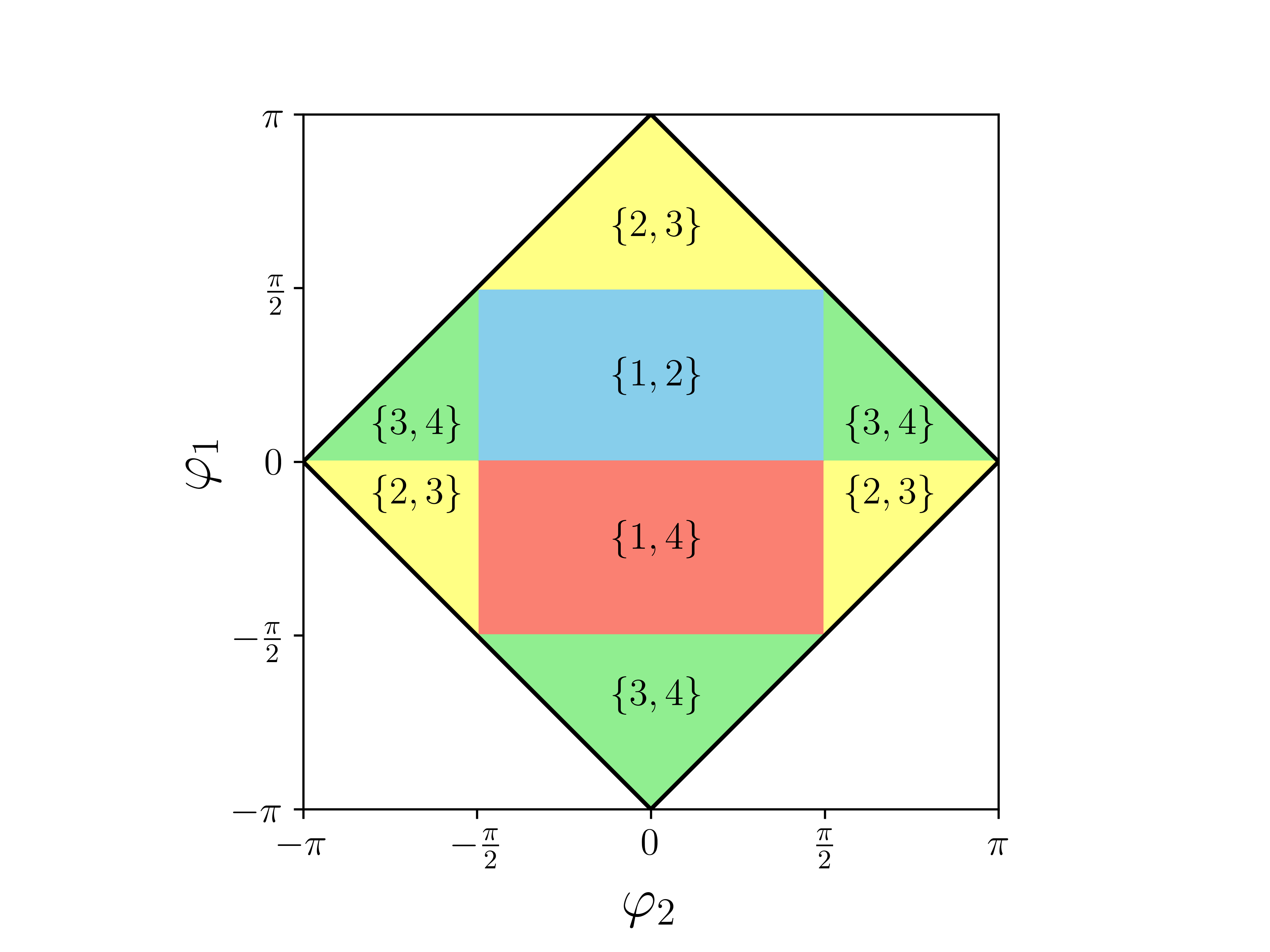}
  \caption{Nearest trajectory in $\mathcal{T}_{\rm cyc}(n=4,m=2)$ as a function of $\varphi_1$ and $\varphi_2$} 
  \label{fig:integration}
\end{figure}

However, computing such an average is challenging because it is difficult to integrate Eq. (\ref{eq:q-succ-exact}) in its exact form. To negotiate this obstacle, we will determine the success probability in the limit of $\theta_0\rightarrow 0$ and $w \gg 1$, which corresponds to the beam being wide and weak. To first order in $\frac{1}{w^2}$, we have
\begin{align}
    \theta_i(\varphi_1,\varphi_2) \approx \theta_0\left(1-\frac{d^2_i(\varphi_1,\varphi_2)}{w^2}\right).
\end{align}
Hence, to first order in $\theta_0$, we have 
\begin{align}
    R_Z(\theta_i(\varphi_1,\varphi_2))\ket{0}&=\exp(-\frac{i\theta_0}{2}\left(1-\frac{d^2_i(\varphi_1,\varphi_2)}{w^2}\right))\ket{0}\\
    &= \exp(-\frac{i\theta_0}{2})\exp(\frac{i\theta_0d^2_i(\varphi_1,\varphi_2)}{2w^2})\ket{0}\\
    &\approx \exp(-\frac{i\theta_0}{2})\left(1+\frac{i\theta_0d^2_i(\varphi_1,\varphi_2)}{2w^2}\right)\ket{0}.
\end{align}
Similarly,
\begin{align}
    R_Z(\theta_i(\varphi_1,\varphi_2))\ket{1}&\approx  \exp(\frac{i\theta_0}{2})\left(1-\frac{i\theta_0d^2_i(\varphi_1,\varphi_2)}{2w^2}\right)\ket{1}.
\end{align}
We now use this result to compute the action of $R_g(\varphi_1,\varphi_2)$ on the basis vectors of $V$:
\begin{align}
    R_g\ket{0011} &\approx \left[1+\frac{i\theta_0}{2w^2}\left(d_1^2+d_2^2-d_3^2-d_4^2\right)\right]\ket{0011}\\
    R_g\ket{0110} &\approx \left[1+\frac{i\theta_0}{2w^2}\left(d_1^2-d_2^2-d_3^2+d_4^2\right)\right]\ket{0110}\\
    R_g\ket{1100} &\approx \left[1-\frac{i\theta_0}{2w^2}\left(d_1^2+d_2^2-d_3^2-d_4^2\right)\right]\ket{1100},\ \mathrm{and}\\
    R_g\ket{1001} &\approx\left[1-\frac{i\theta_0}{2w^2}\left(d_1^2-d_2^2-d_3^2+d_4^2\right)\right]\ket{1001}.
\end{align}
Since $\ket{\psi_{\frac{\pi}{2}}}\in V$ and $\ket{\psi_{\frac{\pi}{2}}^{(T_{\rm min})}}\in V$, we can use the above equations to evaluate Eq. (\ref{eq:q-succ-exact}). Assume that $\varphi_1$ and $\varphi_2$ have been chosen such that $T_{\rm min} = \{1,2\}$. Then $\ket{\psi_{\frac{\pi}{2}}^{(T_{\rm min})}} = \frac{1}{2}(-i\ket{0011}+\ket{0110}+i\ket{1100}+\ket{1001})$, so
\begin{align}
    \bra{\psi_{\frac{\pi}{2}}^{(T_{\rm min})}}R_g(\varphi_1,\varphi_2)\ket{\psi_{\frac{\pi}{2}}}&= \frac{1}{4}\left(i\bra{0011}R_g\ket{0011}+\bra{0110}R_g\ket{0110}-i\bra{1100}R_g\ket{1100}+\bra{1001}R_g\ket{1001}\right)\\
    &\approx \frac{1}{4}\left[2-\frac{\theta_0}{w^2}(d_1^2+d_2^2-d_3^2-d_4^2)\right]\\
    &= \frac{1}{2}-\frac{\theta_0}{4w^2}(d_1^2+d_2^2-d_3^2-d_4^2).
\end{align}
It follows that 
\begin{align}
    \Pr[\mathrm{success}](\varphi_1,\varphi_2) &\approx \frac{1}{4}-\frac{\theta_0}{4w^2}(d_1^2+d_2^2-d_3^2-d_4^2)\\
    &= \frac{1}{4}+\frac{\theta_0}{w^2}(\sin{\varphi_1}+\cos{\varphi_1})\cos{\varphi_2}.
\end{align}
We now compute the average success probability by integrating $\varphi_1$ and $\varphi_2$ over the region where $T_{\rm min} = \{1,2\}$:
\begin{align}
    \operatorname{Pr}_{\rm avg}[\mathrm{success}] &= \frac{2}{\pi^2}\int_{-\frac{\pi}{2}}^{\frac{\pi}{2}}\int_{0}^{\frac{\pi}{2}}d\varphi_1 d\varphi_2\Pr[\mathrm{success}](\varphi_1,\varphi_2)\\
    &\approx \frac{1}{4}+\frac{\theta_0}{w^2}\frac{2}{\pi^2}\int_{-\frac{\pi}{2}}^{\frac{\pi}{2}}\int_{0}^{\frac{\pi}{2}}d\varphi_1 d\varphi_2(\sin{\varphi_1}+\cos{\varphi_1})\cos{\varphi_2}\\
    &= \frac{1}{4}+\frac{\theta_0}{w^2}\frac{2}{\pi^2}\cdot 4\\
    &= \frac{1}{4}+\frac{8}{\pi^2w^2}\theta_0.
\end{align}
Hence, to first order in $\theta_0$ and $\frac{1}{w^2}$, the one-shot failure probability (averaged over possible beam paths) for the quantum beam-estimation protocol is 
\begin{align}
    \operatorname{Pr}_{\rm avg}[\mathrm{failure}]&=1-\operatorname{Pr}_{\rm avg}[\mathrm{success}]\\
    &=\frac{3}{4}-\frac{8}{\pi^2w^2}\theta_0
\end{align}
when $\theta_0\rightarrow 0$ and $w\gg 1$.

Next, we derive the average failure probability of the classical path estimation protocol in the same limit. The probability of measuring qubit $i$ to be in the state $\ket{-}$ after the beam interaction is
\begin{align}
    \abs{\bra{-}R_Z(\theta_i(\varphi_1,\varphi_2))\ket{+}}^2 = \sin^2\left(\frac{\theta_i(\varphi_1,\varphi_2)}{2}\right).
\end{align}
To first order in $\frac{1}{w^2}$ and $\theta_0$, we have
\begin{align}
    \abs{\bra{-}R_Z(\theta_i(\varphi_1,\varphi_2))\ket{+}}^2 &\approx\sin^2\left[\frac{\theta_0}{2}\left(1-\frac{d^2_i(\varphi_1,\varphi_2)}{w^2}\right)\right]\\
    &\approx\frac{\theta_0^2}{4}\left(1-\frac{d^2_i(\varphi_1,\varphi_2)}{w^2}\right)^2\\
    &\approx 0.
\end{align}
Hence, to first order in $\frac{1}{w^2}$ and $\theta_0$, the probability of measuring any of the qubits to be in the state $\ket{-}$ is zero. When none of the qubits are measured to be $\ket{-}$, the classical protocol proposes a trajectory uniformly at random from $\mathcal{T}$. Subsequently, in this limit, the success probability of the classical protocol is $\frac{1}{4}$, and the failure probability is $\frac{3}{4}$.
\bibliography{References}
\bibliographystyle{apsrev4-1}